%% file: SSD_channel_esitmation.tex
\documentclass[12pt, conference, onecolumn]{IEEEtran}
\IEEEoverridecommandlockouts

\usepackage{longtable}
\usepackage{rotating}
\usepackage{multirow}
\usepackage{epsfig, amsmath,amssymb,epsf,cite,subfigure,scalefnt,multirow,array}
\usepackage{algorithm}
\usepackage{algorithmic}

\usepackage{setspace}
\usepackage{epstopdf}
\usepackage{subfig}
\usepackage{amsfonts}
\usepackage{booktabs}
\usepackage{multirow}
\usepackage{multicol}
\usepackage{cite}
\usepackage{algorithm}
\usepackage{algorithmic}
\usepackage{stfloats}
\usepackage{bm}
\usepackage{graphicx}
\usepackage{color}

\input{mydefinitions}

\begin{document}

\title{Sparse Subspace Decomposition for Millimeter Wave MIMO Channel Estimation}
\author{Wei~Zhang$^{\dag}$, Taejoon~Kim$^{\dag}$ and David J. Love$^{\dag\dag}$ \\
\IEEEauthorblockA{ $^{\dag}$\small Department of Electronic Engineering and State Key Laboratory of Millimeter Waves, City University of Hong Kong, Kowloon \\
            \small $^{\dag\dag}$School of Electrical and Computer Engineering, Purdue University, West Lafayette, IN, US \\
            \small Email: wzhang237-c@my.cityu.edu.hk$^{\dag}$, taejokim@cityu.edu.hk$^{\dag}$, djlove@ecn.purdue.edu$^{\dag\dag}$\\}
            \thanks {The work of W. Zhang and T. Kim was fully supported by the Research Grant Council, Hong Kong under Project No. CityU 21200714.} }
\maketitle
\IEEEpeerreviewmaketitle
\begin{spacing}{1.5}
\begin{abstract}
  Millimeter wave multiple-input multiple-output (MIMO) communication systems   must operate over sparse wireless links and will require  large antenna arrays to provide high throughput.  To achieve sufficient array  gains, these systems must learn and adapt to the channel state conditions.
  However,  conventional MIMO channel estimation can not be directly extended to millimeter wave due to the constraints on cost-effective millimeter wave operation imposed on the number of available RF chains.
  Sparse subspace scanning techniques that search for the best subspace sample from the sounded subspace samples have been investigated for channel estimation. 
 However, the performance of these techniques starts to deteriorate as the array size grows, especially for the hybrid precoding architecture.
  The millimeter wave channel estimation challenge still remains and should be properly addressed before the system can be deployed and used to its full potential. In this work, we propose a sparse subspace decomposition (SSD) technique for sparse millimeter wave MIMO channel estimation. We formulate the channel estimation as an optimization problem that minimizes the subspace distance from the received subspace samples. Alternating optimization techniques are devised to tractably handle the non-convex problem. Numerical simulations demonstrate that the proposed method outperforms other existing techniques with remarkably low overhead.
\end{abstract}

\section{Introduction} \label{section1}
The operating frequency of modern wireless systems is steadily shifting to the millimeter wave (e.g.,\! 28-100 GHz) bands which can provide a much wider bandwidth.
These higher frequencies will force systems to use large arrays to generate narrow beams in order to overcome the pathloss and atmospheric impairments encountered at these frequencies. There is particular interest in using advanced multiple antenna concepts at millimeter wave frequencies, but these techniques cannot be directly applied due to the prohibitively complex baseband signal processing overhead\! \cite{RobertOver}.

To alleviate the increased overhead,  hybrid precoding architectures \cite{molisch2004mimo,venka2010analog} have been explored for use at millimeter wave frequencies \cite{nsenga2010mixed, spatially}.
In a hybrid architecture, a large antenna-array is driven by a limited number of RF chains mapped to the antennas using analog processing, i.e., a phase shifter network that linearly processes the RF signal.
Digital processing of the small number of RF chains is then possible.  This architecture, however, does not allow   the digital baseband  to directly access the entire channel, since it only processes analog pre-processed samples.  This  causes a subspace sampling limitation problem \cite{hur}, which  limits the  hybrid system's ability to track and adapt to the channel.

To address the subspace sampling limitation, earlier work mainly focused on hierarchical beam scanning techniques \cite{hur,alk,Wang09}. The transmitter and receiver iteratively search for the best subspace pair by scanning all the subspace pairs chosen in the hierarchical subspace sounding codebooks \cite{hur,alk}. However, the performances of these techniques will be degraded as the antenna size grows, especially when combined with hybrid precoding, due to the increased search overhead.

There have been multiple channel measurement campaigns demonstrating that practical millimeter wave channels present substantially sparse scattering \cite{rappaport,brady}. In this kind of environment, the size of the useful subspace is considerably smaller than the channel dimension. Hence, channel estimation should be adapted  to the sparse channel subspace in order to address the subspace sampling limitation problem.
Recently,  subspace estimation using the Arnoldi method \cite{hadi2015} and the support recovery leveraging the virtual channel representation \cite{kim2015virtual} have been proposed.
With the exception of the work in \cite{hadi2015, kim2015virtual}, little is known on how to efficiently adapt  channel estimation to the sparse but large-dimensional millimeter wave multiple-input multiple-output (MIMO) channel.

In this paper, we formulate  millimeter wave MIMO channel estimation  as a subspace distance minimization problem.
This problem is  related to the general rank minimization problem \cite{semi}, which is NP-hard.
The rank minimization  problem \cite{semi} can be relaxed by using  nuclear norm minimization (NNM)  \cite{recht,chen2015exact}.
Though NNM provides robust performance, it does not guarantee a low rank solution and does not scale to large-dimensional applications due to  computational issues. A matrix factorization (MF) algorithm is proposed in  \cite{Haldar_factor},
which is faster than NNM, to ensure a low rank solution.
The general challenge of the techniques in \cite{recht,chen2015exact,Haldar_factor} is that when noisy observations are made, the reconstruction is often matched to the noise, causing the critical \emph{over-matching} problem.

In this paper, we propose a sparse subspace decomposition technique for the large-dimensional millimeter wave MIMO channel estimation.
The devised technique is general enough to allow channel estimation in  different classes of sparse channels.
We design an alternating optimization technique that tractably handles the formulated sparse decomposition problem.
The algorithm has low complexity and converges to a stationary point.
Moreover, based on intuition from the strong law of large numbers, we develop a simple method to avoid the \emph{over-matching} issue arising in the low signal-to-noise ratio (SNR) regime.
The proposed algorithm outperforms the conventional methods \cite{chen2015exact,Haldar_factor} with substantially reduced computational overhead.


\subsubsection*{Notations}
A bold lower case letter $\ba$ is a vector, a bold capital letter $\bA$ is a matrix. $[\bA]_{i,j}$ is the $i$th row and $j$th column entry of $\bA$, and $[\bA]_i$ is the $i$th column of $\bA$.
$\bA^*,\bA^T,\bA^{\!-1}$, tr($\bA$), $\| \bA \|_F$, $\| \ba \|_2$, $\| \bA \|_*$, respectively, are the
conjuagate transpose, transpose, inverse, trace of $\bA$, Frobenius norm of $\bA$, $\l_2$-norm of $\ba$, and the nuclear norm of $\bA$, which is the sum of its singular values.
$\bI_M \! \in \! \R^{M\times M}$ is the identity matrix.
$\bA  \otimes \bB$ is the Kronecker product of $\bA$ and $\bB$.
Let $\vec(\bA)$ be the operator that stacks columns of $\bA$ into a column vector.
Let $\diag(\bA)$ be the operator that collects the diagonal elements of a square matrix $\bA$ and {forms} a column vector.

\section{System Model} \label{section2}
In this section, we provide signal model, motivations, and general statements of the channel estimation problem in the millimeter wave MIMO systems.
\subsection{Signal Model}
Consider a point-to-point MIMO hybrid precoding system with $N_\mathrm{t}$ transmit and $N_\mathrm{r}$ receive antennas, where each side is equipped with $N_\mathrm{RF}$ RF chains.
The number of data streams is $N_ \mathrm{d}$, and we assume $N_\mathrm{d} \! \leq \! N_\mathrm{RF} \! \leq  \! \min(N_\mathrm{r}, N_\mathrm{t})$.
The transmitter employs the analog precoder $\bF_\mathrm{A} \! \in \! \C^{N_\mathrm{t} \times N_\mathrm{RF}}$ and digital precoder $\bF_\mathrm{D} \! \in \! \C^{N_\mathrm{RF} \times N_\mathrm{d}}$.
At the receiver, $\bW_\mathrm{D} \! \in \!  \C^{N_\text{RF} \times N_\mathrm{d}}$ and $\bW_\mathrm{A} \! \in \!  \C^{N_\mathrm{r} \times N_\mathrm{RF}}$, respectively, denote the digital and analog combiners.
The entries of $\bF_\mathrm{A}$ and $\bW_\mathrm{A}$ are constrained such that $| [\bF_\mathrm{A}]_{i,j} | \!\=\! \frac{1}{\sqrt{N_\mathrm{t}}}$ and $|[\bW_\mathrm{A}]_{i,j} | \!\=\! \frac{1}{\sqrt{N_\mathrm{r}}}$, $\forall i, j$, which are imposed due to analog processing.
The channel input and output relation is therefore given by
\vspace{-0.1cm}
\beq
\textstyle
\br = \bW^*_\mathrm{D} \bW^*_\mathrm{A} \bH \bF_\mathrm{A} \bF_\mathrm{D} \bs + \bW^*_\mathrm{D} \bW^*_\mathrm{A} \bn, \label{receive signal}
\eeq
where $\bH \in \C^{N_\mathrm{r} \times N_\mathrm{t}}$ is the sparse millimeter wave MIMO channel with $\rank(\bH)=L$. We assume $\frac{L}{\min(N_\mathrm{r}, N_\mathrm{t})}\approx 0$.
Notice that, in general, $L$ is not known a priori. Instead, we assume that a rank upper bound of $\bH$, which we define as $d$  (i.e., $\rank({\bH}) \! \leq \! d$), is known a priori.
The vector $\bn\in\C^{N_\mathrm{r} \times 1}$ in \eqref{receive signal} is the additive noise with each entry independent and identically distributed (i.i.d.) as $\bn\!\sim\! \cC\cN(\bzeros, \sig^2 \bI_{N_\mathrm{r}})$.
The vector $\bs\in\C^{N_\mathrm{d} \times 1}$ is the transmit signal satisfying $\mathbb{E} [ \lA \bF_\mathrm{A}\bF_\mathrm{D} \bs \rA_2^2 ] \= 1$. The SNR is, thus, $1/\sig^2$.


Extracting a high quality channel estimate is the key to facilitating advanced MIMO precoding techniques. However, conventional MIMO channel estimation does not directly extend to a millimeter wave hybrid precoding architecture since the digital baseband only accesses the analog compressed channel $\bW_\mathrm{A}^* \bH \bF_\mathrm{A} \! \in \! \C^{N_\mathrm{RF} \times N_\mathrm{RF}}$, rather than the entire $\bH$ as shown in \eqref{receive signal}.
This limitation is the major hurdle when estimating the large-dimensional millimeter wave MIMO channel.
Without direct access to $\bH$, the channel sounding problem is converted to a subspace sampling problem.

\subsection{Subspace Scanning}
CSI acquisition techniques using subspace sampling were investigated in \cite{hur, alk, Wang09}.
Provided a subspace pair $(\bW_{k}, \bF_{k})$, where $\bW_{k}\!\in\!\C^{N_\mathrm{r} \times N_\mathrm{d}}$ and $\bF_{k}\! \in\! \C^{N_\mathrm{t} \times N_\mathrm{d}}$, the subspace sample  $\by_k \! \!\in \! \C^{N_\mathrm{d} \times 1}$\! at the $k$th channel use is
\beq
\by_k = \bW_k^* \bH \bF_k\bs_k + \bW_k^* \bn_k, \ k=1, \ldots, K,  \label{channel sounding}
\eeq
where $\bs_k \in \C^{N_\mathrm{d} \times 1}$ is typically an all one vector with proper normalization.
Each subspace pair $(\bW_{k}, \bF_{k})$ is chosen from a pre-designed subspace sampling codebook \cite{hur, alk, Wang09}.
In particular, if both the digital and analog parts are utilized, we have $\bF_{k} \= \bF_{\mathrm{A},k} \bF_{\mathrm{D},k}$ and $\bW_{k} \= \bW_{\mathrm{A},k} \bW_{\mathrm{D},k}$, $k=1, \ldots, K$.

After collecting $K$ subspace samples, the best pair $(\bW_{k_{opt}}, \bF_{k_{opt}})$  is selected by maximizing the received power,
\vspace{-0.3cm}
\beq
k_{opt} = \argmax_{k=1,\ldots, K} \lA \by_k \rA_2^2.   \label{max power criterion}
\eeq
\vspace{-0.3cm}

\noindent
The best pair $(\bW_{k_{opt}}, \bF_{k_{opt}})$ can be directly used or further processed to generate the precoder and combiner for data transmission. This framework is {investigated in} \cite{hur, Wang09} for beamforming and \cite{alk} for precoding.

\subsection{Motivations and General Statement of Technique}
Obviously, the decesion in \eqref{max power criterion} is highly susceptible to noise.
When the sampling SNR is low, the noise realizations can lead to incorrect subspace decisions.
This is because the subspace scanning approach only looks at the largest power of $\{\by_k\}_{k=1}^K$ rather than correlating the observation with the aligned subspace pairs.
Since the millimeter wave channel is sparse \cite{rappaport,brady}, advanced approaches must also leverage this fact to enhance the channel estimation.
{Our proposed approach will therefore take sparse subspace information into account by
taking the subspace correlation with the noisy observations.}

To this end, we slightly modify the subspace sampling framework in \eqref{channel sounding} to make full use of $N_\mathrm{RF}$ subspace dimension at baseband.
At the $k$th channel use, we have
\beq
\textstyle
\by_k = \bW^*_k \bH \bff_k + \bW^*_k \bn_k,~k=1,2,\ldots, K, \label{new sampling}
\eeq
where $\by_k \! \in \! \C^{N_\mathrm{RF} \times 1}$, $\bW_k \! \in \! \C^{N_\mathrm{r} \times N_\mathrm{RF}}$, and $\bff_k \! \in \! \C^{N_\mathrm{t} \times 1}$.
Unlike \eqref{channel sounding}, the receiver in \eqref{new sampling} collects $N_\mathrm{RF}$ subspace samples per channel use using the subspace pair {$(\bW_k, \bff_k)$} where $\bW_{k} \= \bW_{\mathrm{A},k} \bW_{\mathrm{D},k}$ and {$\bff_{k} \= \bF_{\mathrm{A},k} \bF_{\mathrm{D},k} \bs_k$} with $\bW_{\mathrm{D},k} \! \in \! \C^{N_\mathrm{RF} \! \times \! N_\mathrm{RF}}$ and $\bF_{\mathrm{D},k} \! \in \! \C^{N_\mathrm{RF} \! \times \!  N_\mathrm{RF}}$.
After collecting $K$ subspace samples, we form $\by = [\by_1^T,\by_2^T,\ldots,\by_K^T]^T \! \in   \C^{(K \! N_\mathrm{RF}) \times 1}$.



\subsubsection*{Channel Use Overhead $K$}
The subspace sampling in (4) must be done within a channel coherence block $T$.
Typically, we assume the number of channel uses $K$ such that $K \! \ll \! T$.
Let the minimum required channel uses be $K_{\min}$.
Assuming the noiseless scenario, $K_{\min}$ that guarantees the perfect reconstruction of $\bH \! \in \! \C^{N_\mathrm{r} \times N_\mathrm{t}}$ with $\rank(\bH) \! \leq \! d$ is given by  $K_{\min} \! = \! d(N_\mathrm{r} + N_\mathrm{t} - d)$ [15].
Since the sampling in \eqref{new sampling} can collect in total $K \! N_\mathrm{RF}$ samples, we have the minimum channel uses $K_{\min} \! = \! \frac{d(N_\mathrm{r} + N_\mathrm{t} - d)}{N_\mathrm{RF}}$.
Now, considering the noise, the channel use overhead should satisfy
\beq
K > {d(N_\mathrm{r}+N_\mathrm{t}-d)} / {N_\mathrm{RF}}. \label{minimal samples}
\eeq

\section{Sparse Reconstruction {Objective} Functions} \label{section3}
In this section, we formulate the sparse subspace estimation problem as an optimization problem under the sparse channel constraints.
This problem will be solved in the next section.

First off, suppose the precoding and combining gain maximization problem for  $\bH$ with $\rank(\bH) \leq d$, i.e.,
\beq
\!\! \big( \widehat{\bW}, \widehat{\bF} \big) \!\=  \argmax\limits _{\bW, \bF} \! \lA  \bW^* \bH \bF  \rA_F^2 \ \!  \text{s.t.} \ \bW^* \bW  \! \= \bI_{N_\mathrm{d}},  \bF^* \bF \!\= \bI_{N_\mathrm{d}}, \!\!\!\! \label{precoding gain}
\eeq
where $\bW \in \C^{N_\mathrm{r} \times N_\mathrm{d}}$ and $\bF \in \C^{N_\mathrm{t} \times N_\mathrm{d}}$.
The \eqref{precoding gain} can be written in an alternative form given by
\beq
\begin{aligned}
\big( \widehat{\bU},\widehat{\bV} \big) =&\argmin_{\bU, \bV} \min_{\bLam} \lA \bH - \bU \bLam \bV^* \rA_F^2 \\
&\text{~subject to}~~~\bU^* \bU = \bI_d,~\bV^* \bV = \bI_d,
\label{eq_decom1}
\end{aligned} \label{channel recons}
\eeq
where $\bU \! \in \! \C^{N_\mathrm{r} \times d}$ and $\bV \! \in \! \C^{N_\mathrm{t} \times d}$. The $\bLam \! \in \! \C^{d \times d}$ is a diagonal matrix with the diagonal elements being sorted in  descending order. The relation between \eqref{precoding gain} and \eqref{channel recons} is that $[\widehat{\bW}]_i \!= \! [\widehat{\bU}]_i$ and $[\widehat{\bF}]_i\!=\![\widehat{\bV}]_i$, $i\!=\!1,2,\cdots,N_\mathrm{d}$.
The formulation in \eqref{channel recons} is important because it shows that maximizing the precoding gain in \eqref{precoding gain} is equivalent to approximating the low rank matrix $\bH$ with  $\bU \bLam \bV^*$.
One challenge is how to efficiently use the $K$ subspace samples in \eqref{new sampling} to make a correct approximation.

Inspired by \eqref{eq_decom1}, we now formulate our objective function under the noisy observation that minimizes the subspace distance from the received $K$ subspace samples in \eqref{new sampling}.
Mathematically,
\beq
\begin{aligned}
\big( \widehat{\bU},\widehat{\bV},\widehat{\bLam}\big) = &\argmin_{\bU,\bV,\bLam} \sum_{k = 1}^K \lA \by_k - \bW_k^* \bU \bLam \bV^*\bff_k  \rA^2_2 \\
&\text{~subject to}~~~\bU^* \bU = \bI_d,~\bV^* \bV = \bI_d. \label{objective_tmp}
\end{aligned}
\eeq
The channel estimate will be  $\widehat{\bH} =\widehat{\bU} \widehat{\bLam} \widehat{\bV}^*$.
However, dealing with the summation in \eqref{objective_tmp} is not convenient.
To make the problem intuitive, we introduce an affine map $\cA:  \cF  \times \cW  \times   \C^{N_\mathrm{r} \times N_\mathrm{t}} \mapsto \C^{K \!N_\mathrm{RF}\times 1}$, where $\cF = \{ \bff_k\}_{k=1}^K$ and $\cW = \{ \bW_k\}_{k=1}^K$.
Then, $\cA(\cF,\! \cW,\! \bH)$ yields
\beq
\textstyle
\cA(\cF,\! \cW,\! \bH) \= \!
\left[ (\bW_1^* \bH \bff_1\!)^T \!,  \ldots , (\bW_K^* \bH \bff_K\!)^T   \right]^T \! \!\! \in  \! \C^{K\! N_\mathrm{RF}  \times  1}. \!\!\!\! \label{affine map}
\eeq
The problem in \eqref{objective_tmp} is now rewritten by
\beq
\begin{aligned}
\big( \widehat{\bU},\widehat{\bV},\widehat{\bLam} \big) = &\argmin_{\bU,\bV,\bLam} \lA \by  -
\cA(\cF, \cW, \bU\bLam\bV^*)   \rA^2_2  \\
&\text{~subject to}~~~\bU^* \bU = \bI_d,~\bV^* \bV = \bI_d. \label{eq_objective}
\end{aligned}
\eeq


The problem in \eqref{eq_objective} has much overlap with the rank minimzation problems \cite{semi,recht,chen2015exact,Haldar_factor}.
The original rank minimization problem \cite{semi} is, however, non-convex and is infeasible to be directly solved (i.e., NP-hard). One alternative approach is to use the nuclear norm heuristic \cite{recht}. The nuclear norm minimization (NNM) problem leveraging the millimeter wave channel subspace samples in \eqref{new sampling} can be written as
\beq
\begin{aligned}
\widehat{\bH} = \argmin  \limits_ {\bH}  \lA \bH \rA_* ~~\text{subject to~~}  \by \!-\!  \cA(\cF, \cW,\bH) \in \cC. \label{nuclear norm} \nonumber
\end{aligned}
\eeq
The set $\cC$ is a convex set which can be adjusted based on the noise power \cite{chen2015exact}.
This problem is convex and solvable.

Though NNM provides robust performance, it does not guarantee the low rank solution, leading to inaccurate channel estimate under the noisy observation.
The most significant drawback is that the computational complexity becomes prohibitively high as $N_\mathrm{r}$ and $N_\mathrm{t}$ grow.

Much like our setting, the matrix factorization (MF) technique \cite{Haldar_factor} is proposed to guarantee the low rank solution and can be formulated as
\beq
\begin{aligned}
\big( \widehat{\bU}, \widehat{\bV} \big) = \argmin \limits_ {\bU,\bV} \lA \by - \cA(\cF,\cW, \bU \bV) \rA_2^2,  \label{matrix factor}
\end{aligned}
\eeq
where $\bU \! \in \! \C^{N_\mathrm{r} \times d}$ and $\bV \! \in \! \C^{N_\mathrm{t} \times d}$.
An alternating optimization can be used to find a local minimum of \eqref{matrix factor}. The computation complexity is lower than NNM.
However, since $\bU$ and $\bV$ are not constrained to lie on the unitary subspace, when it is compared with \eqref{eq_objective} and NNM,
it exhibits the worst performance (as shown in Section \ref{section5}).
%

\section{Sparse Millimeter Wave Channel Estimation} \label{section4}
Our method will address two major challenges: low-rank guarantee and reduced complexity with robust performance.
\subsection{Sparse Subspace Decomposition (SSD)} \label{section4A}

The objective function in \eqref{eq_objective} finds the {column and row} subspaces of $\bH$, i.e., $\bU$ and $\bV$, respectively, along with the power allocation matrix $\bLam$ to meet $\bU^* \bU = \bI_{d}$ and $\bV^* \bV = \bI_{d}$.
Unfortunately, these semi-unitary constraints are not tractable to handle since they are non-convex.
Though \emph{Cayley Transform} \cite{Zaiwen} can be employed to deal with the orthogonality constraints, the computational complexity is prohibitively high.
Instead, we will consider the convex relaxation of \eqref{eq_objective} as
\beq
\begin{aligned}
\big( \widehat{\bU},\widehat{\bV},\widehat{\bLam} \big) = &\argmin \limits_ {\bU,\bLam,\bV} \lA \by - \cA(\cF,\cW, \bU \bLam \bV^*) \rA_2^2
\\
&\text{subject to~~} \tr(\bU^* \bU)  \le d,~\tr(\bV^* \bV )  \le d.
\label{SVD method_relax}
\end{aligned}
\eeq

Notice that the problem \eqref{SVD method_relax} is still non-convex since those optimization parameters are coupled each other.
Nevertheless, the problem \eqref{SVD method_relax} can be suboptimally but tractably solved by using alternating minimization techniques.
In particular, optimizing one parameter by fixing the other two parameters in \eqref{SVD method_relax} is convex.
We can iteratively optimize $\bU$, $\bV$, and $\bLam$ by solving the following {three subproblems}:
\begin{itemize}
\item[(S1)] Fix the row subspace $\bV$ and power allocation $\bLam$, optimize the column subspace $\bU$,
\item[(S2)] Fix the column subspace $\bU$ and power allocation $\bLam$, optimize the row subspace $\bV$,
\item[(S3)] Fix the row subspace $\bV$ and column subspace $\bU$, optimize the power allocation $\bLam$.
\end{itemize}

\noindent A formal description of the alternating minimization is provided in Algorithm \ref{alg1}.
Since (S1), (S2), and (S3) are all convex, we can obtain the global optima of each subproblem.
Moreover, due to the covexity, the {objective} function $\| \by - \cA(\cF,\cW,\bU \bLam \bV^*) \|_2^2$ converges over the iterations.

Now, the optimal solutions of the subproblems are of interest. We limit our discussion to the column subspace optimization problem (S1),
bearing in mind that the same applies to the row subspace optimization (S2).
For simplicity, we omit the iteration index $\ell$ attached to the variables throughout this subsection.
The following lemma provides the solution to (S1).

\begin{algorithm} [t]
\caption{Sparse subspace decomposition (SSD) for millimeter wave channel estimation}
\label{alg1}
\begin{algorithmic} [1]
\STATE Input: Codebooks $\cF$ and $\cW$, and subspace samples $\by$.
\STATE Initialization: Set iteration number $\ell=0$. $\bU_{(0)}$, $\bV_{(0)}$, and $\bLam_{(0)}$ are arbitrary.
\REPEAT
\STATE Update the column subspace of $\bH$:
\beq
\! \! \! \! \! \! \! \! \! \! \! \! \! \! \! \! \! \! \! \! \! \! \! \! \!    \text{(S1)}
\begin{cases}
 \bU_{(\ell+1)} \! = \! & \d4 \argmin \limits_ {\bU} \big\| \by  \!-\! \cA(\cF,\cW,\bU \bLam_{(\ell)} \bV_{(\ell)}^*) \big\|_2^2\\
&\! \! \! \! \! \! \text{subject to~}  \tr(\bU^* \bU) \le d,
\label{SVD method U}
\end{cases}
\eeq

\STATE Update the row subspace of $\bH$:
\beq
\! \! \! \! \! \! \! \!  \! \! \! \! \! \! \! \! \! \!  \text{(S2)}
\begin{cases}
 \bV_{(\ell+1)} \! = \! &\! \! \! \! \! \! \argmin \limits_ {\bV} \lA \by \!-\! \cA(\cF,\cW,\bU_{(\ell+1)} \bLam_{(\ell)} \bV^*) \rA_2^2\\
&\! \! \! \! \! \! \text{subject to~}  \tr(\bV^* \bV) \le d,
\label{SVD method V}
\end{cases}
\eeq

\STATE Update the power allocation to the subspaces $\bLam$:
\beq
\! \! \! \! \! \! \! \!  \! \! \! \!  \! \!  \!  \! \!\text{(S3)~}
\! \bLam_{(\ell+1)}  \! = \! \argmin \limits_ {\bLam}  \! \big \| \by  \!-\! \cA(\cF,\cW,\bU_{(\ell+1)} \bLam \bV_{(\ell+1)}^*) \big\|_2^2,
\label{SVD method lamd}
\eeq
\STATE$\ell  =  \ell + 1,~\text{and}~{\widehat{\bH}}_{(\ell)} =\bU_{(\ell)} \bLam_{(\ell)} \bV_{(\ell)}^*$ \label{update}
\UNTIL{$\ell$ exceeds a maximum number of iterations or the iterations stagnate.}
\STATE Output: $\widehat{\bH} =\widehat{\bH}_{(\ell)}$.
\end{algorithmic}
\end{algorithm}
%
%

\begin{lemma}\label{t1}
Suppose the following quadratic programming,
\beq
\begin{aligned}
\! \widehat{\bU} \= \argmin \limits_ {\bU} \! \lA \by \- \cA(\cF,\cW,\bU \bLam \bV^*) \rA_2^2 ~ \text{s.t.~}  \tr(\bU^* \bU) \!\leq d,
\label{SVD method U omit}
\end{aligned}
\eeq
where $\by\in\C^{K \! N_\mathrm{RF} \times 1}$, $\bU\in \C^{N_\mathrm{r} \times d}$, $\bV\in \C^{N_\mathrm{t} \times d}$, and $\bLam \in \C^{d \times d}$.
Let $\bB$ be
\beq
\bB \! = \! [((\bLam \bV^* \bff_1)^T \otimes  \bW_1^*)^T,\ldots,((\bLam \bV^* \bff_K)^T \otimes  \bW_K^*)^T]^T \label{matrix B}
\eeq
where $\bB \in \C^{(K \! N_\mathrm{RF}) \times (dN_\mathrm{r})}$.
Then, the optimal solution $\widehat{\bU}$ is given by either
$ \vec(\widehat{\bU}) = (\bB^*\bB)^{-1} \bB^* \by$ such that $ \| \vec(  \widehat{\bU} ) \|_2^2  \leq d $
or
$\vec(\widehat{\bU}) \= (\bB^*\bB \+ \lambda \bI_{d N_\mathrm{r}})^{-1} \bB^* \by$ such that  $\ g(\lambda) \! \triangleq \! \| \vec(  \widehat{\bU} ) \|_2^2  \= d$,
where $\lambda > 0$ is the unique solution to the fixed point equation $g(\lambda) \= d$,
in which $g(\lambda)$ is monotonically decreasing in $\lambda>0$\footnote{The $\lambda > 0$  that satisfies $\| \vec({\widehat{\bU}} ) \|_2^2=d$ can be tractably found by using a bisection method.}.
\end{lemma}
\begin{proof}
Given the affine map $\cA$ in \eqref{affine map}, using the Kronecker product equality, $\vec( \bW_k^* \bU \bLam \bV^* \bff_k) = ( (\bLam \bV^* \bff_k)^T \otimes  \bW_k^* ) \vec(\bU)$, for $k=1, \ldots, K$ and collecting them as a column vector yield an equivalent problem to \eqref{SVD method U omit} as
\beq
\widehat{\bU} \! = \! \argmin \limits_ {\bU} \lA \by \! - \! \bB\vec(\bU) \rA_2^2 ~\text{subject to~} \lA \vec( {\bU} ) \rA_2^2  \! \le \! d,
\label{SVD method U new}
\eeq
where $\bB$ is defined in \eqref{matrix B}. This problem is convex and applying KKT condition \cite{boyd} for \eqref{SVD method U new} gives
\beq
\begin{cases}
\bB^* \bB \vec({\bU})-\bB^*\by \+ \lam \vec({\bU}) \= 0, \ \text{(\emph{first-order condition})} \nonumber \\
\lam \ge 0,\ \text{(\emph{dual constraint})} \\
\lam [\lA \vec( {{\bU}} ) \rA_2^2 -d] = 0, \ \text{(\emph{complementary slackness})} \label{KKT}

\end{cases}
\eeq
where $\lam$ is the Lagrange multiplier for the constraint in \eqref{SVD method U new}.
By the \emph{complementary slackness}, the optimal solution is found by considering two cases, i.e.,  $\lam=0$ and $\lA \vec( {\bU} ) \rA_2^2 - d = 0$.

When $\lam \! = \! 0$, we have  $\vec(\widehat{\bU}) \!= \!(\bB^* \bB)^{-1} \bB^* \by$ with $\| \vec( {\widehat{\bU}} ) \|_2^2 \!\le\! d$.
According to \eqref{minimal samples}, $K \!  N_\mathrm{RF} \! \ge \! {d(N_\mathrm{r}+N_\mathrm{t}-d)} \! \ge \! dN_\mathrm{r}$, hence, $\bB^* \bB$ is invertible.

When $\lam \!>\! 0$, $\lA \vec( {\bU} ) \rA_2^2 \!=\!d$ holds and the optimal $\widehat{\bU}$ should meet the two conditions
$\vec(\bU)=(\bB^* \bB + \lam \bI)^{-1} \bB^* \by,~\lam > 0$ and $\lA \vec( {\bU} ) \rA_2^2=d$.
This concludes the proof.\qed
\end{proof}

We turn our attention to the subspace power allocation problem (S3).

\begin{lemma}
Suppose the subspace power allocation problem
\beq
\widehat{\bLam} = \argmin \limits_ {\bLam} \lA \by - \cA(\cF,\cW,\bU \bLam \bV^*) \rA^2_2,
\label{SVD method lamd omit}
\eeq
where $\bLam \! \in \! \C^{d \times d}$ is the diagonal matrix, $\by \! \in \! \C^{K \!N_\mathrm{RF} \times 1}$, $\bU \! \in \! \C^{N_\mathrm{r} \times d}$, and $\bV \! \in \! \C^{N_\mathrm{t} \times d}$.  Then the optimal solution to \eqref{SVD method lamd omit} is
$$ \diag(\widehat{\bLam}) = (\bP^* \bP)^{-1}\bP^* \by,  $$
where $\bP \! \in \! \C^{K \! N_\mathrm{RF} \times d}$, and $[\bP]_j \!= \! \cA(\cF,\cW,[\bU]_j [\bV]_j^*),~\forall j$.
\end{lemma}

\begin{proof}
Decompose the affine map $\cA$ in \eqref{affine map}  as
\beq
\textstyle\d4 \cA(\cF,\cW, \bU \bLam  \bV^*) \d4&=&\d4 \textstyle \cA(\cF,\cW, \sum\limits_{j=1}^d [\bU]_j [\bLam]_{j,j}  [\bV]_j^*) \nn \\
                \d4&=&\d4 \textstyle  \sum\limits_{j=1}^d  [\bLam]_{j,j} \cA(\cF,\cW,[\bU]_j [\bV]_j^*),\label{linear operator}
\eeq
where $[\bLam]_{j,j}$ in \eqref{linear operator} is the power allocation to the corresponding subspace $[\bU]_j [\bV]_j^*$.  Using \eqref{linear operator}, the problem \eqref{SVD method lamd omit} can be rewritten as
\vspace{-0.2cm}
\beq
\textstyle
\widehat{\bLam} = \argmin \limits_ {\bLam} \big\| \by - \sum \limits_ {j=1}^d [\bLam]_{j,j} \cA(\cF,\cW,[\bU]_j [\bV]_j^*) \big\|_2^2.
\label{SVD method lamd new}
\eeq
\vspace{-0.1cm}
Solving the least square problem in \eqref{SVD method lamd new} results in
\beq
\diag(\widehat{\bLam}) = (\bP^* \bP)^{-1}\bP^* \by \label{SVD method lamd closed 2}, \nonumber
\eeq
where $\bP \! \in \! \C^{K \! N_\mathrm{RF} \times d}$, and $[\bP]_j \!= \! \cA(\cF,\cW,[\bU]_j [\bV]_j^*),~j=1,\cdots,d$. This concludes the proof. \qed
\end{proof}


\begin{algorithm} [t]
\caption{Sparse subspace decomposition with thresholding (SSD-T) for millimeter wave channel estimation}
\label{alg2}
\begin{algorithmic} [1]
\STATE Input: Codebooks $\cF$ and $\cW$, subspace samples $\by$, and noise variance $\sig^2$.
\STATE Initialization: Set iteration number $\ell=0$. $\bU_{(0)}$, $\bV_{(0)}$ and $\bLam_{(0)}$ are arbitrary.
\REPEAT
\STATE Get $\bH_{(\ell)}$ from Algorithm 1,
\STATE  $\ell  =  \ell + 1$,
\UNTIL{$\ell$ exceeds a maximum number of iterations, the iterations stagnate, or \eqref{2a} is satisfied.}
\STATE Output: $\widehat{\bH} =\widehat{\bH}_{(\ell)}$.
\end{algorithmic}
\end{algorithm}

\begin{figure}[t]
\centering
\includegraphics[width=2.7in]{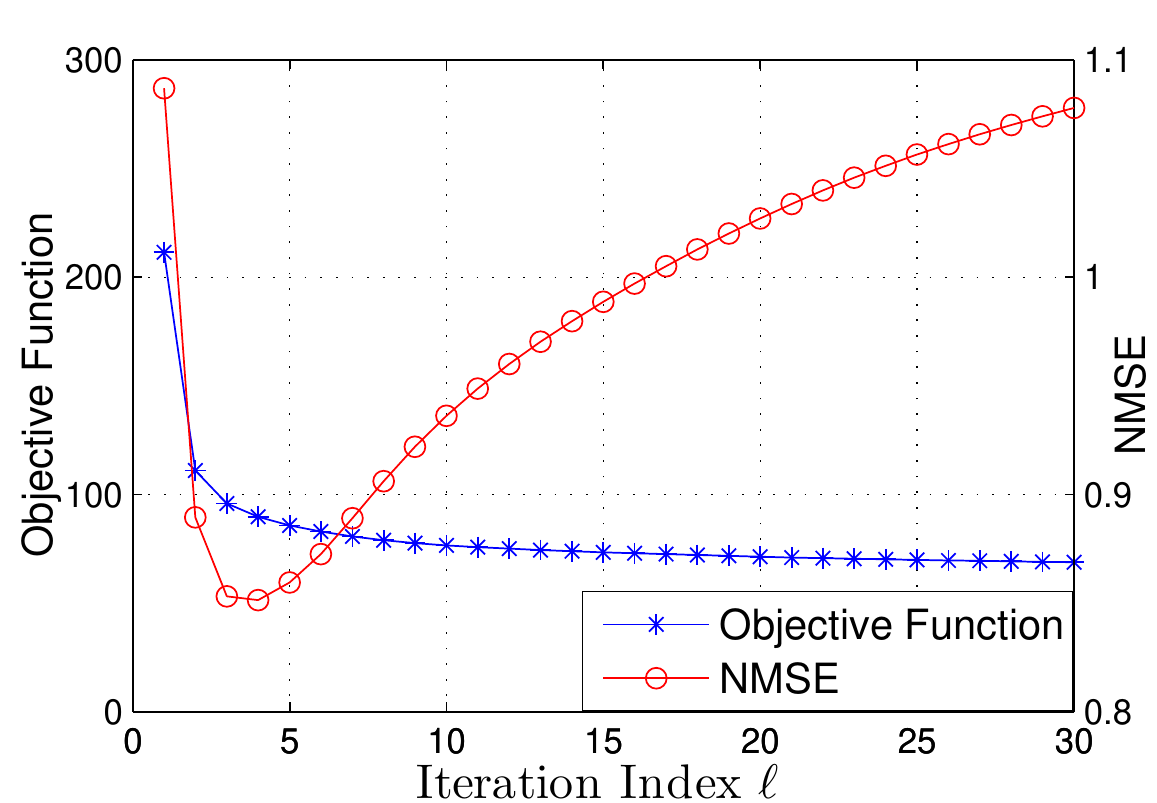}
\caption{Objective function of \eqref{SVD method_relax} and NMSE vs. iteration $\ell$} \label{Fig2}
($N_\mathrm{t}=64, N_\mathrm{r}=16, L=2, K=100, N_\mathrm{RF} =4, \text{SNR} \! = \!0$~dB).
\end{figure}

\subsection{\mbox{\!Sparse Subspace\! Decomposition with\! Thresholding \!(SSD-T)}} \label{section4B}
In our channel estimation problem, a normalized mean square error (NMSE), i.e.,
\beq
\textstyle
\text{NMSE} =  \mathbb{E} \Big[{\big\| \bH - \widehat{\bH} \big\|^2_F}/ {\lA \bH \rA^2_F}\Big]. \label{metric}
\eeq
is evaluated to confirm the accuracy of an estimator.
Obviously, there is a gap between the objective function in \eqref{SVD method_relax} and NMSE in \eqref{metric}.
Though our objective function monotonically decreases per iteration to a local minimum point, the NMSE in \eqref{metric} would exhibit a non-monotonic behavior.
This is because the objective in \eqref{SVD method_relax} contains the disturbance, i.e., the noise $\bn$.
The decoupling between the objective and NMSE becomes critical in a low SNR, especially when $\| \by \|_2^2 \approx \| \tilde{\bn} \|_2^2$, where $\tilde{\bn} = [(\bW_1^* \bn_1)^T , \cdots, (\bW_K^* \bn_K)^T]^T$.
In this case, the objective function is approximated as
$\| \tilde{\bn} \- \cA(\cF,\cW,\widehat{\bH}) \|_2^2$.
Minimizing $\| \tilde{\bn} \- \cA(\cF,\cW,\widehat{\bH}) \|_2^2$ results in large deviation from the true $\bH$.

To illustrate this in a low SNR scenario, we display the curves of the objective function of \eqref{SVD method_relax} and NMSE of the proposed SSD (i.e., Algorithm \ref{alg1}) across iteration $\ell$ in Fig. \ref{Fig2} with $0$ dB SNR.
As expected, the objective function of \eqref{SVD method_relax} decreases monotonically.
However, NMSE rather increases after a few iterations because the algorithm finds $\widehat{\bH}$ close to the noise, resulting in  the {\emph{over-matching}}.

One remedy is to terminate the iteration rather than making $\| \by - \cA(\cF,\cW,\widehat{\bH}) \|_2^2$ as small as possible.
Specifically, provided the output of Algorithm \ref{alg1} at iteration $\ell$, i.e., $\widehat{\bH}_{(\ell)}$, the stopping condition can be given by
\beq
\textstyle
{\big\| \by - \cA(\cF,\cW, \widehat{\bH}_{(\ell)}) \big\|_2^2}\big/(K\! N_\mathrm{RF}) < \sig^2. \label{2a}
\eeq
The underlying intuition of \eqref{2a} is that by the strong law of large number the left hand side (l.h.s.) of \eqref{2a} will converge to the variance of noise, i.e., $\sig^2$, as $K$ grows, given that the output of Algorithm 1 is close to the true $\bH$. The inequality in \eqref{2a} signifies that we treat the case when the l.h.s. of \eqref{2a} is smaller than $\sig^2$ as the \emph{over-matching} and stop the iteration.
The formal description of the SSD combined with the thresholding in \eqref{2a},
namely SSD with thresholding (SSD-T),
is provided in Algorithm \ref{alg2}.
By reducing the number of iterations, Algorithm \ref{alg2} will  decrease the computational complexity while enhancing the performance in a low SNR.

\section{Numerical Simulations} \label{section5}

In this section, we numerically evaluate the NMSE of the proposed SSD and SSD-T algorithms and compare them with the matrix factorization (MF) \cite{Haldar_factor} and nuclear norm minimization (NNM) \cite{chen2015exact} techniques.
Also, we benchmark the computational overheads of SSD, SSD-T, MF, and NNM.
The numerical simulation setting is first discussed.

(1) \emph{Channel model}:
We assume the prevalent physical channel representation  that models sparse millimeter wave MIMO channels \cite{RobertOver, brady}.
Assume there are $L$  propagation paths between the transmitter and receiver, where the channel is modeled via
\beq
\bH = \sqrt{\frac{N_\mathrm{r} N_\mathrm{t}}{L}} \sum\limits_{l = 1}^L \al_l\ba_\mathrm{r}({\phi_{\mathrm{r},l}}) \ba^*_\mathrm{t} ({\phi_{\mathrm{t},l}}). \nonumber
\eeq
The $\al_{l}$ is the complex gain of the $l$th path, i.i.d. as $\al_{l}\sim \cC\cN(0, 1)$, $\phi_{\mathrm{t},l}$ and $\phi_{\mathrm{r},l}$ are angles of departure and arrival at the transmitter and receiver, respectively.
The $\ba_\mathrm{t}(\phi_{\mathrm{t},l})$ and {$\ba_\mathrm{r}(\phi_{\mathrm{r},l})$} are the uniform linear array (ULA) response vectors at the transmitter and receiver, respectively, where the inter-element spacing of ULA is set to half of the wavelength.
The maximum number of iteration of the proposed SSD (i.e., Algorithm \ref{alg1}) is 30.
We further assume $L\=2$ paths, $N_\mathrm{RF} \=4$ RF chains, and the number of data stream $N_\mathrm{d}\=L$. We assume the rank upper bound $d\=3$.

(2) \emph{Performance Evaluation}:
The NMSE statistics across different SNRs and channel uses ($K$) are evaluated.
Each curve is obtained after averaging over $100$ channel realizations.

\begin{table}[t]
\caption{Execution Time (in seconds) for Algorithms.}
\centering
\begin{tabular}{c c | c c c c c}
\hline
$N_\mathrm{r} \times N_\mathrm{t}$ & Channel \\Uses ($K$) & MF & NNM & SSD & SSD-T \\
\hline
$16\times 64$ & 100 & 5.23 & 228.83 & 15.09 & 2.41  \\
$16\times 64$ & 150 & 10.60 & 483.09 & 25.05 & 3.78\\
$8\times 16$  & 20 & 0.32  & 3.47 & 1.19 & 0.17\\
$8\times 16$  & 60 & 1.20    & 12.60  &3.91 & 0.57\\ \hline
\end{tabular} \label{table1}
\end{table}

Table \ref{table1} shows the execution time statistics (in seconds on a $3.4$ GHz CPU) for $8 \times 16$ and $16 \times 64$ channel dimensions when SNR = 10 dB.
It can be seen from Table \ref{table1} that both proposed algorithms exhibit lower complexity.
The SSD-T shows the fastest execution while the NNM becomes substantially slower, as the antenna dimension grows.

Fig. \ref{Fig3} illustrates the NMSE curves of the proposed SSD and SSD-T compared with MF and NNM for $16\times 64$ channel and $K\=100$ channel uses.
It can be seen that both SSD and SSD-T outperform the MF across a broad range of SNRs. NNM achieves robust performance at the expense of the huge computational overhead.
On the other hand, Fig. \ref{Fig3} reveals that SSD-T achieves similar performance as NNM in the low SNR regime, while it starts to outperform NNM in the high SNR, verifying that SSD-T has the ability to resolve the complexity and accuracy tradeoffs.

Fig. \ref{Fig4} presents NMSE curves at $10$ dB and $30$ dB SNRs across different channel uses.
Fig. \ref{Fig4} clearly shows that both SSD and SSD-T achieve similar performances with NNM at $10$ dB SNR, while keeping the complexity remarkably low.
At $30$ dB SNR, both SSD and SSD-T outperform NNM.
The gap of the proposed approaches, compared with MF, is clear in Fig. \ref{Fig4}.
Evidently, increasing $K$ in the high SNR has the effect of enhancing the accuracy of SSD and SSD-T, while only marginal improvement is observed for NNM.

\begin{figure}[t]
\centering
\includegraphics[width=3in]{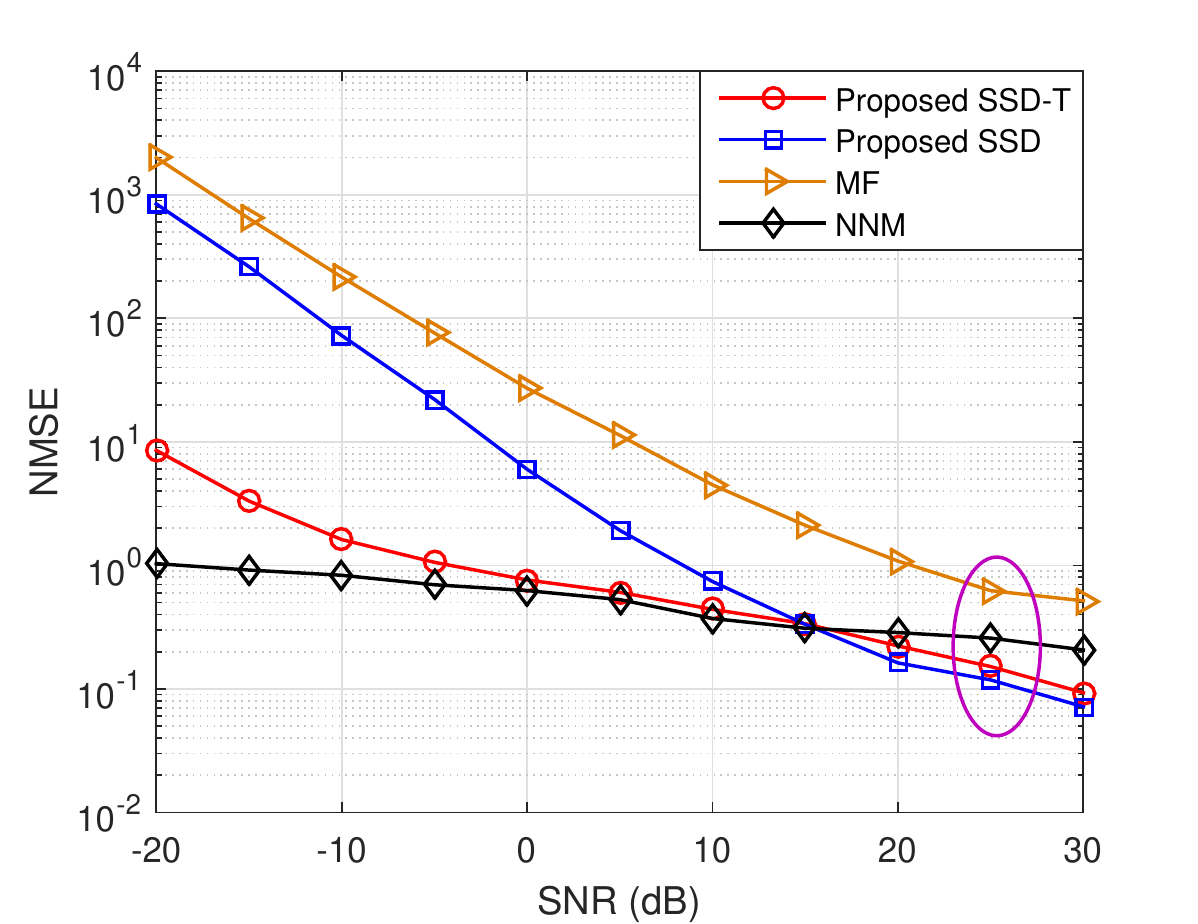}
\caption{NMSE vs. SNR (dB) ($N_\mathrm{t}=64, N_\mathrm{r}=16, L=2, K=100, N_\mathrm{RF} =4$).} \label{Fig3}
\end{figure}

\section{Conclusions} \label{section6}
In this paper, we investigated the estimation of sparse millimeter wave MIMO channels   using hybrid precoding systems.
Based on the equivalency of optimizing the unitary precoder/combiner of the sparse channel matrix and approximating the channel via a sparse subspace decomposition, we proposed the SSD algorithm.
The SSD consists of three constituent estimators: one for the row subspace, one for the column subspace, and the other for the subspace power allocation.
To avoid  unexpected over-matching issues, we devised the SSD-T algorithm.
Our simulations showed that SDD and SDD-T achieve similar or better accuracy than the best performing existing sparse matrix estimation benchmarks (e.g., NNM algorithm) with remarkably low overhead.


\begin{figure}[t]
\centering
\includegraphics[width=3 in]{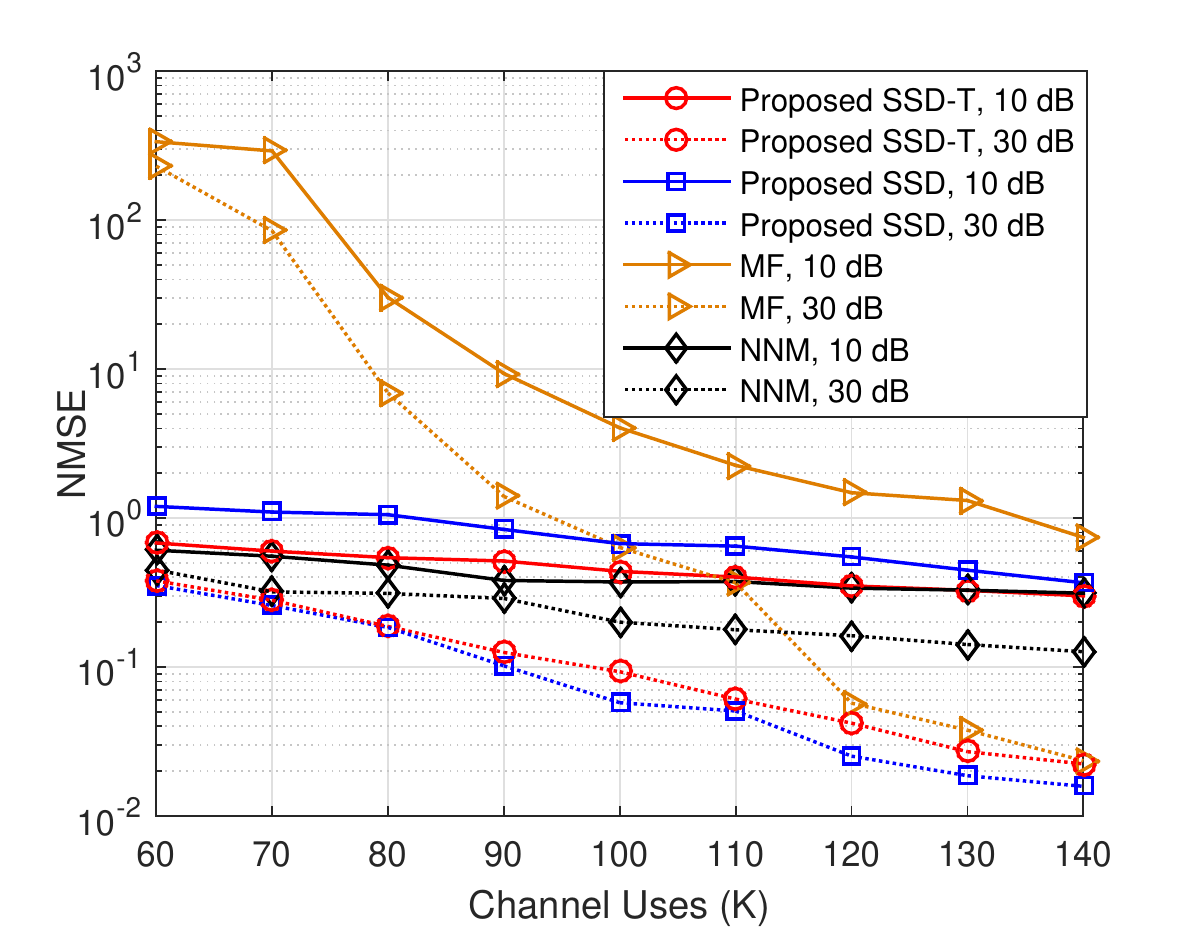}
\caption{NMSE vs. Channel Uses ($N_\mathrm{t}=64, N_\mathrm{r}=16, L=2, \text{SNR =10~dB and 30~dB}, N_\mathrm{RF} =4$).} \label{Fig4}
\end{figure}



\bibliographystyle{IEEEtran}
\bibliography{IEEEabrv,Conference_mmWave_CS}
\end{spacing}

\end{document}

%% file: mydefinitions.tex
\def\cA{{\mathcal{A}}}  \def\cC{{\mathcal{C}}} 
 \def\cF{{\mathcal{F}}}  
   
 \def\cN{{\mathcal{N}}}  
   
  \def\cW{{\mathcal{W}}}

\def\ba{{\mathbf{a}}}    
\def\bff{{\mathbf{f}}}    
   \def\bn{{\mathbf{n}}} 
  \def\br{{\mathbf{r}}} \def\bs{{\mathbf{s}}} 
    \def\by{{\mathbf{y}}}

\def\bA{{\mathbf{A}}} \def\bB{{\mathbf{B}}}   
\def\bF{{\mathbf{F}}}  \def\bH{{\mathbf{H}}} \def\bI{{\mathbf{I}}} 
    
\def\bP{{\mathbf{P}}}    
\def\bU{{\mathbf{U}}} \def\bV{{\mathbf{V}}} \def\bW{{\mathbf{W}}}

\def\argmin{\mathop{\mathrm{argmin}}}
\def\argmax{\mathop{\mathrm{argmax}}}

\def\tr{\mathop{\mathrm{tr}}}

\def\rank{\mathop{\mathrm{rank}}}
\def\span{\mathop{\mathrm{span}}}
\def\vec{\mathop{\mathrm{vec}}}

\def\diag{\mathop{\mathrm{diag}}}

\def\bzeros{\mathbf{0}}

\renewcommand{\baselinestretch}{1.0}

     \def\d4{\!\!\!\!}              
                       \def\al{\alpha}
               
   \def\bLam{\mathbf{\Lambda}}
    \def\lam{\lambda}

  \def\R{{\mathbb{R}}} \def\C{{\mathbb{C}}}

           \def\lA{\left\|}     \def\rA{\right\|}

  \def\sig{\sigma}

  \def\-{\! - \!}  \def\+{\! + \!}  \def\={\! = \!}  \def\>{\! > \!} \def\nn{\nonumber}

\newtheorem{lemma}{Lemma}

\newcommand{\bef}{\begin{figure}}
\newcommand{\eef}{\end{figure}}
\newcommand{\beq}{\begin{eqnarray}}
\newcommand{\eeq}{\end{eqnarray}}

\newenvironment{proof}[1][Proof]{\begin{trivlist}
\item[\hskip \labelsep {\bfseries #1}]}{\end{trivlist}}

\newcommand{\qed}{\nobreak \ifvmode \relax \else
\ifdim\lastskip<1.5em \hskip-\lastskip \hskip1.5em plus0em
minus0.5em \fi \nobreak \vrule height0.5em width0.5em
depth0.25em\fi}

%% file: SSD_channel_esitmation.bbl
\begin{thebibliography}{10}
\providecommand{\url}[1]{#1}
\csname url@samestyle\endcsname
\providecommand{\newblock}{\relax}
\providecommand{\bibinfo}[2]{#2}
\providecommand{\BIBentrySTDinterwordspacing}{\spaceskip=0pt\relax}
\providecommand{\BIBentryALTinterwordstretchfactor}{4}
\providecommand{\BIBentryALTinterwordspacing}{\spaceskip=\fontdimen2\font plus
\BIBentryALTinterwordstretchfactor\fontdimen3\font minus
  \fontdimen4\font\relax}
\providecommand{\BIBforeignlanguage}[2]{{%
\expandafter\ifx\csname l@#1\endcsname\relax
\typeout{** WARNING: IEEEtran.bst: No hyphenation pattern has been}%
\typeout{** loaded for the language `#1'. Using the pattern for}%
\typeout{** the default language instead.}%
\else
\language=\csname l@#1\endcsname
\fi
#2}}
\providecommand{\BIBdecl}{\relax}
\BIBdecl

\bibitem{RobertOver}
R.~W. Heath, N.~Gonzalez-Prelcic, S.~Rangan, W.~Roh, and A.~M. Sayeed, ``An
  overview of signal processing techniques for millimeter wave \text{MIMO}
  systems,'' \emph{IEEE Journal of Selected Topics in Signal Processing},
  vol.~10, no.~3, pp. 436--453, April 2016.

\bibitem{molisch2004mimo}
A.~F. Molisch and M.~Z. Win, ``\text{MIMO} systems with antenna selection,''
  \emph{Microwave Magazine, IEEE}, vol.~5, no.~1, pp. 46--56, 2004.

\bibitem{venka2010analog}
V.~Venkateswaran and A.-J. Van~der Veen, ``Analog beamforming in \text{MIMO}
  communications with phase shift networks and online channel estimation,''
  \emph{Signal Processing, IEEE Transactions on}, vol.~58, no.~8, pp.
  4131--4143, 2010.

\bibitem{nsenga2010mixed}
J.~Nsenga, A.~Bourdoux, and F.~Horlin, ``Mixed analog/digital beamforming for
  60 \text{GHz MIMO} frequency selective channels,'' in \emph{Communications
  (ICC), 2010 IEEE International Conference on}.\hskip 1em plus 0.5em minus
  0.4em\relax IEEE, 2010, pp. 1--6.

\bibitem{spatially}
O.~El~Ayach, S.~Rajagopal, S.~Abu-Surra, Z.~Pi, and R.~W. Heath, ``Spatially
  sparse precoding in millimeter wave \text{MIMO} systems,'' \emph{Wireless
  Communications, IEEE Transactions on}, vol.~13, no.~3, pp. 1499--1513, 2014.

\bibitem{hur}
S.~Hur, T.~Kim, D.~J. Love, J.~V. Krogmeier, T.~A. Thomas, and A.~Ghosh,
  ``Millimeter wave beamforming for wireless backhaul and access in small cell
  networks,'' \emph{Communications, IEEE Transactions on}, vol.~61, no.~10, pp.
  4391--4403, 2013.

\bibitem{alk}
A.~Alkhateeb, O.~El~Ayach, G.~Leus, and R.~W. Heath, ``Channel estimation and
  hybrid precoding for millimeter wave cellular systems,'' \emph{Selected
  Topics in Signal Processing, IEEE Journal of}, vol.~8, no.~5, pp. 831--846,
  2014.

\bibitem{Wang09}
J.~Wang, Z.~Lan, C.~Pyo, T.~Baykas, C.~Sum, M.~A. Rahman, J.~Gao, R.~Funada,
  F.~Kojima, H.~Harada, and S.~Kato, ``Beam codebook based beamforming protocol
  for multi-{G}bps millimeter-wave {WPAN} systems,'' \emph{IEEE Journal on
  Selected Areas in Communications}, vol.~27, no.~8, pp. 1390--1399, Oct 2009.

\bibitem{rappaport}
T.~S. Rappaport, S.~Sun, R.~Mayzus, H.~Zhao, Y.~Azar, K.~Wang, G.~N. Wong,
  J.~K. Schulz, M.~Samimi, and F.~Gutierrez, ``Millimeter wave mobile
  communications for 5{G} cellular: It will work!'' \emph{Access, IEEE},
  vol.~1, pp. 335--349, 2013.

\bibitem{brady}
J.~Brady, N.~Behdad, and A.~M. Sayeed, ``Beamspace \text{MIMO} for
  millimeter-wave communications: System architecture, modeling, analysis, and
  measurements,'' \emph{Antennas and Propagation, IEEE Transactions on},
  vol.~61, no.~7, pp. 3814--3827, 2013.

\bibitem{hadi2015}
H.~Ghauch, T.~Kim, M.~Bengtsson, and M.~Skoglund, ``Subspace estimation and
  decomposition for large millimeter-wave \text{MIMO} systems,'' \emph{IEEE
  Journal of Selected Topics in Signal Processing}, vol.~10, no.~3, pp.
  528--542, April 2016.

\bibitem{kim2015virtual}
T.~Kim and D.~J. Love, ``Virtual {AoA} and {AoD} estimation for sparse
  millimeter wave \text{MIMO} channels,'' in \emph{Signal Processing Advances
  in Wireless Communications (SPAWC), 2015 IEEE 16th International Workshop
  on}.\hskip 1em plus 0.5em minus 0.4em\relax IEEE, 2015, pp. 146--150.

\bibitem{semi}
L.~Vandenberghe and S.~Boyd, ``Semidefinite programming,'' \emph{SIAM review},
  vol.~38, no.~1, pp. 49--95, 1996.

\bibitem{recht}
B.~Recht, M.~Fazel, and P.~A. Parrilo, ``Guaranteed minimum-rank solutions of
  linear matrix equations via nuclear norm minimization,'' \emph{SIAM review},
  vol.~52, no.~3, pp. 471--501, 2010.

\bibitem{chen2015exact}
Y.~Chen, Y.~Chi, and A.~J. Goldsmith, ``Exact and stable covariance estimation
  from quadratic sampling via convex programming,'' \emph{Information Theory,
  IEEE Transactions on}, vol.~61, no.~7, pp. 4034--4059, 2015.

\bibitem{Haldar_factor}
J.~Haldar and D.~Hernando, ``Rank-constrained solutions to linear matrix
  equations using powerfactorization,'' \emph{Signal Processing Letters, IEEE},
  vol.~16, no.~7, pp. 584--587, July 2009.

\bibitem{Zaiwen}
Z.~Wen and W.~Yin, ``A feasible method for optimization with orthogonality
  constraints,'' \emph{Mathematical Programming}, vol. 142, no. 1-2, pp.
  397--434, 2013.

\bibitem{boyd}
S.~Boyd and L.~Vandenberghe, \emph{Convex optimization}.\hskip 1em plus 0.5em
  minus 0.4em\relax Cambridge university press, 2004.

\end{thebibliography}
